\documentclass[a4paper,10pt]{article}
\usepackage[latin1]{inputenc}
\usepackage[T1]{fontenc}
\usepackage{amssymb}
\usepackage{epsfig}
\usepackage{amssymb}
\usepackage{amsmath}
\usepackage{amsfonts}
\usepackage{graphicx, url}
\usepackage{makeidx}
\usepackage{cite}
\usepackage{fancyhdr}
\usepackage{color}
\usepackage{setspace}

\newtheorem{theorem}{Theorem}[section]

\newtheorem{lemma}[theorem]{Lemma}
\newtheorem{remark}[theorem]{Remark}

\newtheorem{example}[theorem]{Example}

\newenvironment {proof} {{\it Proof.}}{\hspace*{\fill}$\Box$\par\vspace{4mm}}
\usepackage{amsmath}
\usepackage{amsfonts}

\newcommand{\mc}{\mathcal}
\newcommand{\mb}{\mathbb}

\newcommand{\cl}{\mbox{\rm cl}}

\newcommand{\co}{\mbox{\rm co}}

\begin{document}

\title{A convex duality approach for pricing contingent claims under partial information and short selling constraints}

\author{Kristina Rognlien Dahl \footnote{Department of Mathematics, University of Oslo. kristd@math.uio.no. The research leading to these results has received funding from the European Research Council under the European Community's Seventh Framework Programme (FP7/2007-2013) / ERC grant agreement no. 228087.}}

\maketitle

\begin{abstract}

We consider the pricing problem facing a seller of a contingent claim. We assume that this seller has some general level of partial information, and that he is not allowed to sell short in certain assets. This pricing problem, which is our primal problem, is a constrained stochastic optimization problem. We derive a dual to this problem by using the conjugate duality theory introduced by Rockafellar. Furthermore, we give conditions for strong duality to hold. This gives a characterization of the price of the claim involving martingale- and super-martingale conditions on the optional projection of the price processes.





\end{abstract}

\section{Introduction}
\label{sec: introduction}

This paper analyzes an optimization problem from mathematical finance using conjugate duality. We consider the pricing problem of a seller of a contingent claim $B$ in a discrete time, arbitrary scenario space setting. The seller has a general level of partial information, and is subject to short selling constraints. The seller's (stochastic) optimization problem is to find the minimum price of the claim such that she, by investing in a self-financing portfolio, has no risk of losing money at the terminal time $T$. The price processes are only assumed to be non-negative, stochastic processes, so the framework is model independent (in this sense).

The main contribution of the paper is a characterization of the dual of the seller's price of the claim $B$ as a $Q$-expectation of the claim, where $Q$ is a mixed martingale- and super-martingale measure with respect to the conditional expectation of the price process, see Theorem~\ref{thm: omskriving}. To the best of our knowledge, this is a new result. The mix of martingale- and super-martingale measure is due to the presence of short selling constraints on some of the assets, while the conditional expectation is due to the seller's partial information. The optimal value of this dual problem is an upper bound of the seller's price. To prove this characterization, we use a conjugate duality technique. This technique is different from what is common in the mathematical finance literature, and results in (fairly) brief proofs. Moreover, it does not rely on the reduction to a one-period model. This feature makes it possible to solve the optimization problem even though it contains partial information.

Conjugate duality (also called convex duality), which is used to analyze the seller's problem, is a general framework for studying and solving optimization problems. This framework was introduced by Rockafellar~\cite{Rockafellar}, see Appendix \ref{sec: conjugate} for a brief summary. For a further treatment of conjugate duality and its role in stochastic optimization, see Shapiro et al.~\cite{Shapiro}.

Some of the main features of this paper are:

\begin{itemize}
\item{We have a completely general filtration representing partial information. This is in contrast to for instance Kabanov and Stricker~\cite{KabanovStricker}, where delayed information is used.}
\item{The use of conjugate duality is a general approach and it provides an efficient way of deriving the dual of the seller's pricing problem, without reduction to a one-period model.}
\item{Since we use discrete time, general price processes are considered.}
\end{itemize}

The use of conjugate duality in mathematical finance is a fairly recent development. Over the last few years, Pennanen has done some pioneering work in this area, see Pennanen~\cite{Pennanen}, \cite{Pennanen2}, \cite{Pennanen3} as well as Pennanen and Perkki\"{o}~\cite{PP}. King~\cite{King} and King and Korf~\cite{KingKorf} have also worked on the connection between conjugate duality and mathematical finance. Duality theory in a broader sense is at the core of mathematical financial theory. Various kinds of duality, such as linear programming duality, Lagrange duality and the bipolar theorem, are used in many areas of finance. For instance, Pinar~\cite{Pinar1},\cite{Pinar2} applies Lagrange duality to derive dual representations for contingent claim pricing using a gain-loss criterion. In the setting of the present paper, this Lagrange duality approach is equivalent to our conjugate duality method. However, conjugate duality has the advantage that it can be generalized to a continuous time setting as well. In particular, duality theory (typically, in infinite dimensions) is used in utility maximization, hedging, analyzing convex risk measures, consumption and investment problems and optimal stopping. Kramkov and Schachermayer~\cite{KramkovSchachermayer1}, \cite{KramkovSchachermayerIgjen}, Karatzas and Shreve~\cite{KaratzasShreve} and Pham~\cite{Pham} consider duality in utility maximization problems. The books by Karatzas and Shreve~\cite{KaratzasShreve} and Pham~\cite{Pham} also consider duality in hedging. Pliska~\cite{Pliska} uses linear programming duality in arbitrage-related problems. Frittelli and Rozassa Gianin~\cite{Frittelli} apply conjugate duality to convex risk measures. Also, Rogers considers many applications of duality in mathematical finance, for instance in  consumption, investment and hedging problems, see Rogers~\cite{Rogers1} as well as Klein and Rogers~\cite{KleinRogers}. Rogers also derives a pure dual method for solving optimal stopping problems, see Rogers~\cite{Rogers2}.

For more on replication of claims under short selling constraints, see
Cvitani{\`c} and Karatzas~\cite{CK}, F\"{o}llmer and Kramkov~\cite{FK}, Jounini and Kallal~\cite{JouiniAndKallal}, Karatzas and Kou~\cite{KK}, Karatzas and Shreve~\cite{KaratzasShreve} and Pulido~\cite{Pulido}.

Kabanov and Stricker~\cite{KabanovStricker} derive a version of the Dalang-Morton-Willinger theorem under delayed information. They do this by generalizing a proof of the no-arbitrage criteria from Kabanov et al.~\cite{KSY}. Their result is related to our pricing result Theorem~\ref{thm: price2}, in the sense that it involves martingale conditions on the optional projection of the price processes. However, in contrast to Kabanov and Stricker~\cite{KabanovStricker}, we have completely general partial information (i.e., it does not need to be delayed information). Moreover, we consider pricing of claims, not arbitrage problems like in \cite{KabanovStricker}. We also have short-selling constraints and our methods, in particular the use of conjugate duality, are different than those in ~\cite{KabanovStricker}. Bouchard~\cite{Bouchard} and De Valli{\'e}re et al.~\cite{DKS} also consider no arbitrage conditions under partial information, but with transaction cost, and without short-selling constraints like we do.

The rest of the paper is organized as follows: Section~\ref{sec: shorting2} introduces the financial market model and analyzes the seller's optimization problem by deriving a dual problem using conjugate duality. Section~\ref{sec: main} consists of our main theorem with proof, and gives an alternative characterization of the dual problem. In Section~\ref{sec: strongduality} it is shown that there is no duality gap in the case without borrowing or short-selling. By combining this with the previous results, we find a characterization of the seller's price involving martingale- and super-martingale conditions. We also give a numerical example to illustrate the results. Finally, Section~\ref{sec: finalremarks} concludes, and poses some open questions for further research.

\section{Pricing with short selling constraints and partial information}
\label{sec: shorting2}

We model the financial market as follows. There is a given probability space $(\Omega, \mc{F}, P)$ consisting of a scenario space $\Omega$, a $\sigma$-algebra $\mc{F}$ on $\Omega$ and a probability measure $P$ on the measurable space $(\Omega, \mc{F})$. The financial market consists of $N+1$ assets: $N$ risky assets (stocks) and one non-risky asset (a bond). The assets each have a (not identically equal zero) stochastic price process $S_{n}(t, \omega)$, $n = 0, 1, \ldots, N$, for $\omega \in \Omega$ and $t \in \{0, 1, \ldots, T\}$ where $T < \infty$, and $S_0$ denotes the price process of the bond. We denote by $S(t,\omega) := (S_{0}(t,\omega), S_1(t,\omega), \ldots, S_N(t,\omega))$, the vector in $\mb{R}^{N+1}$ consisting of the price processes of all the assets. We assume that $S_0(t,\omega) := 1$ for all $t \in \{0, 1, \ldots,T\}$, $\omega \in \Omega$, so the market is discounted. Let $(\mc{F}_t)_{t=0}^T$ be a filtration corresponding to full information in the market. We assume that the price process $S$ is adapted to this filtration. For more on a similar kind of framework, see {\O}ksendal~\cite{Oksendal}.

 Associated with each seller in the market there is a filtration $(\mc{G}_t)_t := (\mc{G}_t)_{t=0}^{T}$, where $\mc{G}_0 = \{\emptyset, \Omega\}$ and $\mc{G}_T = \mc{F}$. The filtration represents the development of the information available to the seller. The assumptions on $\mc{G}_0$ and $\mc{G}_T$ imply that at time $0$ the seller knows nothing, while at time $T$ the true world scenario is revealed. We assume that $\mc{G}_t \subseteq \mc{F}_t$ for all $t=0, 1, \ldots, T$. This means that the seller only has partial information, in contrast to Kabanov and Stricker~\cite{KabanovStricker}, where they use delayed information. By considering a general partial information, we include for instance the possibility of unobserved/hidden processes for the seller.

Let $H_n(t,\omega)$, $n=0, 1, \cdots, N$ be the number of units of asset number $n$ the seller has at time $t \in \{0, 1, \ldots, T-1\}$ in scenario $\omega \in \Omega$. Then, the seller chooses a trading strategy
\[
H(t,\omega) := (H_0(t,\omega), H_1(t,\omega), \cdots, H_N(t,\omega))
\]
\noindent based on this information. Since the seller at each time chooses this trading strategy based on her current information, $H(t)$ is $\mc{G}_t$-measurable for all $t \in \{0, 1, \ldots, T-1\}$. Hence, the trading strategy process $(H(t))_{t \in \{0, 1, \ldots, T-1\}}$ is $(\mc{G}_t)_{t}$-adapted. Let the space of all such $(\mc{G}_t)_t$-adapted trading strategies $H$ be denoted by $\mc{H}_{\mc{G}}$.

We consider the pricing problem of a seller of a non-negative $\mc{F}$-measurable contingent claim $B$ ($B$ is non-negative without loss of generality by translation). Let $I_1 \subseteq \{1, 2, \ldots, N\}$ be a subset of the risky assets, and let $I_2 = \{1, 2, \ldots, N\} \setminus I_1$ (i.e., the compliment of $I_1$). The seller is not allowed to short sell in risky asset $S_j$, where $j \in I_1$. Also, we assume that there is no arbitrage w.r.t. $(\mc{G}_t)_t$.

Let $\Delta H(t) := H(t) - H(t-1)$. The seller's optimization problem is:
\begin{equation}
\label{eq: forsteligning}
\begin{array}{lll}
\inf_{\{v, H\}} \hspace{1cm} v \\[\smallskipamount]
\mbox{subject to} \\[\smallskipamount]
\begin{array}{rcll}
(i) \mbox{ } S(T) \cdot H(T-1) &\geq& B &\mbox{ a.s.}, \\[\smallskipamount]
(ii) \hspace{0.56cm} \mbox{ } S(t) \cdot \Delta H(t) &=& 0 &\mbox{ for } 1 \leq t \leq T-1, \mbox{ a.s.}, \\[\smallskipamount]
(iii) \hspace{1.62cm} \mbox{ } H_{j}(t) &\geq& 0 &\mbox{ for } 0 \leq t \leq T-1, \mbox{ a.s.}, j \in I_1\\[\smallskipamount]
(iv) \hspace{0.77cm} \mbox{ } S(0) \cdot H(0) &\leq& v,
\end{array}
\end{array}
\end{equation}
\noindent where $v \in \mb{R}$ and $H$ is $(\mc{G}_t)_t$-adapted. Note that the inequality $(iii)$ is the no short-selling constraint. Hence, the seller's problem is: Minimize the price $v$ of the claim $B$ such that the seller is able to pay $B$ at time $T$ (constraint $(i)$) from investments in a self-financing (constraint $(ii)$), adapted (w.r.t. the partial information) portfolio that costs less than or equal to $v$ at time $0$ (constraint $(iv)$). In addition, the trading strategy cannot involve selling short in assets $S_j$, $j \in I_1$ (constraint $(iii)$). Note that if there is a $(\mc{G}_t)_t$-arbitrage, problem~\eqref{eq: forsteligning} is unbounded. Also, the absence of arbitrage under the full information filtration $(\mc{F}_t)_t$ implies absence of arbitrage under the partial information $(\mc{G}_t)_t$.

Note that problem~\eqref{eq: forsteligning} is an infinite linear programming problem, i.e. the problem is linear with infinitely many constraints and variables. For more on infinite programming, see for instance Anderson and Nash~\cite{AndersonNash} and for a numerical method, see e.g. Devolder et al.~\cite{DevolderEtAl}.  However, if $\Omega$ is finite, \eqref{eq: forsteligning} is a linear programming problem. In this case, the problem can be solved numerically using the simplex algorithm or an interior point method, see for example Vanderbei~\cite{Vanderbei}

We will rewrite problem~\eqref{eq: forsteligning} in a way suitable for determining its dual. Clearly, one can remove constraint $(iv)$, and instead minimize over $S(0) \cdot H(0)$. Also, since there is no $(\mc{G}_t)_t$-arbitrage, it suffices to minimize over the portfolios such that $S(0) \cdot H(0) \geq 0$. Then, the pricing problem is a minimization problem with four types of constraints ($S(0) \cdot H(0) \geq 0$ is the fourth type). Now, the problem can be rewritten so it fits the conjugate duality framework (see Appendix~\ref{sec: conjugate} for a general presentation of conjugate duality or Rockafellar~\cite{Rockafellar}). Let $|I_1|$ denote the number of elements in $I_1$, that is the number of assets the seller is not allowed to short-sell in. Let $p \in [1,\infty)$ and the perturbation space $U$ be defined by
\[
\begin{array}{lll}
U := \{u = (\gamma,(w_t)_{t=1}^{T-1}, (x_t^{(j)})_{t=0, \mbox{ } j \in I_1}^{T-1},z) : u \in \mc{L}^p(\Omega, \mc{F}, P: \mb{R}^{(|I_1| + 1)T + 1})\}.
\end{array}
\]
\noindent Define (for notational convenience) $w := (w_t)_{t=1}^{T-1}$ and $x^{(j)} :=  (x_t^{(j)})_{t=0}^{T-1}$.

Let $Y := U^* = \mc{L}^q(\Omega, \mc{F}, P: \mb{R}^{(|I_1| + 1)T + 1})$, the dual space of $U$, where $\frac{1}{p} + \frac{1}{q} = 1$. Note that $y := (y_1, (y_2^t)_{t=1}^{T-1}, (\xi_t^{(j)})_{t=0, j \in I_1}^{T-1}, y_3) \in Y$ has components corresponding to $u \in U$. Note also that $u$ consists of four types of variables, $\gamma, w, (x^{(j)})_{j \in I_1}$, and $z$. Each of these variables correspond to a constraint type in the rewritten minimization problem. The same will hold for the dual variable $y$. Consider the pairing of $U$ and $Y$ using the bilinear form

\[
\begin{array}{lll}
\langle u,y \rangle = \mb{E}[u \cdot y].
\end{array}
\]

Choose the perturbation function $F : \mc{H}_{\mc{G}} \times U \rightarrow \mb{R}$ (again, see Appendix~\ref{sec: conjugate} for more on perturbation functions) in the following way:
\begin{enumerate}
 \item[($i$)] If $B - S(T) \cdot H(T-1) \leq \gamma$ a.s., $S(t) \cdot \Delta H(t) = w_t$ for all $t \in \{1, \ldots, T-1\}$ a.s., $-H_{j}(t) \leq x_t^{(j)}$ for all $t \in \{0, \ldots, T-1\}$, $j \in I_1$ a.s., $S(0) \cdot H(0) \geq z$, then let $F(H, u) :=S(0) \cdot H(0)$.
\item[($ii$)] Otherwise, let $F(H, u) := \infty$.
\end{enumerate}

The corresponding Lagrange function is
\[
\begin{array}{lll}
K(H, y) &=&S(0) \cdot H(0)+ \mb{E}[y_1 (B - S(T) \cdot H(T-1))] \\[\smallskipamount]
                    &&+\sum_{t=1}^{T-1} \mb{E}[y_2^t S(t) \cdot \Delta H(t)] - \sum_{j \in I_1} \sum_{t=0}^{T-1} \mb{E}[\xi_t^j H_{j}(t)] - \mb{E}[y_3 S(0) \cdot H(0)]
\end{array}
\]
\noindent if $y_1, \xi_t^j, y_3 \geq 0$ a.s. for all $t \in \{0, \ldots, T-1\}$ and $K(H, y) = -\infty$ otherwise. We can now determine the (conjugate) dual problem to the primal problem~\eqref{eq: forsteligning}. By collecting terms for each $H_i(t)$, the dual objective function is
\begin{equation}
\label{eq: snartdual}
\begin{array}{rllll}
g(y) :=& &\inf_{\{H \mbox{ : } (\mc{G}_t)_t-\mbox{adapted}\}} K(H,y) \\[\smallskipamount]
=& &\mb{E}[y_1 B]  
+ \sum_{i \in I_2} \inf_{H_i(0)} \{\mb{E}[H_i(0) \{S_i(0)(1-y_3) - y_2^1 S_i(1)\}]\} \\[\smallskipamount]
&&+ \sum_{j \in I_1} \inf_{H_{j}(0)} \{\mb{E}[H_{j} \{S_{j}(0)(1-y_3) - y_2^1 S_{1}(1) - \xi_0^{(j)}\}]\} \\[\smallskipamount]
&&+ \sum_{t=1}^{T-2} \big( \sum_{i \in I_2} \inf_{H_i(t)}\{\mb{E}[H_i(t)(y_2^t S_i(t) 
- y_2^{t+1} S_i(t+1))]\}\\[\smallskipamount]
&&+\sum_{j \in I_1} \inf_{H_{j}(t)} \{\mb{E}[H_{j}(t)(y_2^t S_{j}(t) - y_2^{t+1}S_{j}(t+1) - \xi_t^{(j)})]\} \big) \\[\smallskipamount]
&&+ \sum_{i \in I_2} \inf_{H_i(T-1)} \{\mb{E}[H_i(T-1)(-y_1 S_i(T) + y_2^{T-1} S_i(T-1))]\} \\[\smallskipamount]
&&+ \sum_{j \in I_1} \inf_{H_{j}(T-1)} \{\mb{E}[H_{j}(T-1)(-y_1 S_{j}(T) + y_2^{T-1}S_{j}(T-1) - \xi_{T-1}^{(j)})]\}.
\end{array}
\end{equation}

\subsection{Two Lemmas}
\label{sec: lemmas}

This section consists of two lemmas needed in the following presentation. We include the proofs for completeness.

\begin{lemma}
\label{lemma: prisingKonjugert}
Let $f$ be any random variable w.r.t. $(\Omega, \mc{F}, P)$ and let $\mc{G}$ be a sub-$\sigma$-algebra of $\mc{F}$. Let $\mc{X}$ denote the set of all $\mc{G}$-measurable random variables. Then
\[
\inf_{\{g \in \mc{X}\}} \mb{E}[fg] > -\infty
\]
if and only if $\int_A f dP = 0$ for all $A \in \mc{G}$.
\end{lemma}

\begin{proof}
$\Rightarrow$: Assume there exists $A \in \mc{G}$ such that $\int_A f dP = K \neq 0$. Define $g(\omega) := M$ for all $\omega \in A$, where $M$ is a constant, and $g(\omega) := 0$ for all $\omega \in \Omega \setminus A$. The result follows by letting $M \rightarrow +/- \infty$

$\Leftarrow$: Prove the result for simple functions. The Lemma follows by an approximation argument.

\end{proof}

In the next lemma the notation is the same as in Lemma~\ref{lemma: prisingKonjugert}:

\begin{lemma}
\label{lemma: prisingKonjugert2}
$\inf_{\{g \in \mc{X}\}} \mb{E}[fg] > -\infty$ implies that $\inf_{\{g \in \mc{X}\}} \mb{E}[fg] = 0$.
\end{lemma}

\begin{proof}
Follows by observing that $\inf_{\{g \in \mc{X}\}} \mb{E}[fg] \leq 0$ ($g = 0$ is feasible) and the definition of the infimum.

\end{proof}

By combining Lemma~\ref{lemma: prisingKonjugert2} with Lemma~\ref{lemma: prisingKonjugert}, it follows that $\inf_{\{g \in \mc{X}\}} \mb{E}[fg] = 0$ if and only if $\int_A f dP = 0$ for all $A \in \mc{G}$.

\medskip

There exists a feasible dual solution if and only if all the infima in equation~(\ref{eq: snartdual}) are greater than $-\infty$. To derive the dual problem, we consider each of these minimization problems separately and use the comment after Lemma~\ref{lemma: prisingKonjugert} and Lemma~\ref{lemma: prisingKonjugert2}. We also use that since $\xi_t^{(j)} \geq 0$ a.e. for all $t \in \{0, 1, \ldots, T-1\}$ and $j \in I_1$, then $\int_A \xi_{t}^{(j)} dP \geq 0$ for all $A \in \mc{G}_t$ for all $t$. Also, from the derived dual feasibility conditions, it is sufficient to only maximize over solutions where $y_3=0$ $P$-a.e. Note that such a solution exists, because we have assumed that there is no $(\mc{G}_t)_t$-arbitrage. Hence, the dual problem is
\begin{equation}
\label{eq: opprdualshortvilkomega2}
\begin{array}{lll}
\sup_{\{y \in Y: y_1 \geq 0\}} \hspace{1cm} \mb{E}[y_1 B] \\[\smallskipamount]
\mbox{s.t.} \\[\smallskipamount]
\begin{array}{lrcll}
    (i) &\int_A S_i(0) dP &=& \int_A y_2^1 S_i(1) dP &\forall \mbox{ } A \in \mc{G}_0, \\[\smallskipamount]
    (i)^* &\int_A S_{j}(0) dP &\geq& \int_A y_2^1 S_{j}(1) dP &\forall \mbox{ } A \in \mc{G}_0, \\[\smallskipamount]
    (ii) &\int_A y_2^t  S_i(t) dP &=& \int_A y_2^{t+1} S_i(t+1) dP &\forall \mbox{ } A \in \mc{G}_t, t = 1, \ldots, T-2, \\[\smallskipamount]
    (ii)^* &\int_A S_{j}(t) y_2^t dP &\geq& \int_A y_2^{t+1} S_{j}(t+1) dP &\forall \mbox{ } A \in \mc{G}_t, t = 1, \ldots, T-2, \\[\smallskipamount]
    (iii) &\int_A y_2^{T-1} S_i(T-1) dP &=& \int_A y_1 S_i(T) dP &\forall \mbox{ } A \in \mc{G}_{T-1}, \\[\smallskipamount]
    (iii)^* &\int_A y_2^{T-1} S_{j}(T-1) dP &\geq& \int_A y_1 S_{j}(T)dP &\forall \mbox{ } A \in \mc{G}_{T-1}
\end{array}
\end{array}
\end{equation}
\noindent where the equality constraints $(i), (ii)$ and $(iii)$ hold for $i \in I_2$ and the inequality constraints $(i)^*, (ii)^*$ and $(iii)^*$ hold for $j \in I_1$. Note that the dual feasibility conditions come in pairs, where the only difference is whether there is $=$ (short selling allowed) or $\geq$ (short selling not allowed).

This dual problem~\eqref{eq: opprdualshortvilkomega2} is, like the primal problem~\eqref{eq: forsteligning}, an infinite linear programming problem. As before, if $\Omega$ is finite, it is a regular linear programming problem which can be solved using the simplex algorithm or an interior point method. However, this version of the dual problem is not significantly simpler to solve than the original problem. Therefore, we will rewrite problem~\eqref{eq: opprdualshortvilkomega2} in a more interpretable form, which in some cases is more attractive to solve than the primal problem.

\section{The main theorem}
\label{sec: main}

In this section, we will show our main theorem, Theorem~\ref{thm: omskriving}, which states that the dual problem~(\ref{eq: opprdualshortvilkomega2}) is equivalent to another problem involving martingale- and super-martingale conditions on the optional projection of the price process.

In the following, let $\bar{\mc{M}}_{I_1}^a(S,\mc{G})$ be the set of probability measures $Q$ on $(\Omega, \mc{F})$ that are absolutely continuous w.r.t. $P$ and are such that the price processes $S_i$ for $i \in I_2$ satisfy $E_Q[S_i(t+k) | \mc{G}_t] = E_Q[S_i(t) | \mc{G}_t]$, while for $j \in I_1$ they satisfy $E_Q[S_j(t+k) | \mc{G}_t] \leq E_Q[S_j(t) | \mc{G}_t]$ for $k \geq 0$ and $t \in 0,1, \ldots, T-k$, i.e. $Q$ is a mixed martingale and super-martingale measure for the optional projection of the price process.

\begin{theorem}
\label{thm: omskriving}
The dual problem~(\ref{eq: opprdualshortvilkomega2}) is equivalent to the following optimization problem.
\begin{equation}
\label{eq: nydualshortvilkomega}
\begin{array}{llll}
\sup_{Q \in \bar{\mc{M}}_{I_1}^a(S,\mc{G})} \mb{E}_Q[B].
\end{array}
\end{equation}
\end{theorem}

\begin{proof}
First, assume there exists a $Q \in \bar{\mc{M}}_{I_1}^a(S,\mc{G})$, i.e., a feasible solution to problem~(\ref{eq: nydualshortvilkomega}). We want to show that there is a corresponding feasible solution to problem~(\ref{eq: opprdualshortvilkomega2}).

Define $y_1 := \frac{dQ}{dP}$ (the Radon-Nikodym derivative of $Q$ w.r.t. $P$, see Shilling~\cite{Shilling}), and $y_2^t := \mb{E}[y_1 | \mc{F}_t]$ for $t = 0, 1, \ldots, T-1$. We prove that $y_1, y_2^t$ satisfy the dual feasibility conditions of problem (\ref{eq: opprdualshortvilkomega2}).

\begin{itemize}
\item{$(iii)^*$: 
    From the definition of conditional expectation, it suffices to prove
    \[
    \int_A \mb{E}[y_1 S_{j}(T) | \mc{G}_{T-1}] dP \leq \int_A y_2^{T-1} S_{}(T-1) dP \mbox{ for all } A \in \mc{G}_{T-1}, j \in I_1.
    \]
    In particular, it suffices to prove
    \[
    \mb{E}[y_1 S_{j}(T) | \mc{G}_{T-1}] \leq \mb{E}[y_2^{T-1} S_{j}(T-1) | \mc{G}_{T-1}] \mbox{ } P\mbox{-a.e}.
    \]
    By the definition of $y_2^{T-1}$, this is equivalent to
    \[
    \mb{E}[y_1 S_{j}(T) | \mc{G}_{T-1}] \leq \mb{E}[\mb{E}[y_1 | \mc{F}_{T-1}] S_{j}(T-1) | \mc{G}_{T-1}] \mbox{ } P\mbox{-a.e}.
    \]
    Since $S_{j}(T-1)$ is $\mc{F}_{T-1}$-measurable, the inequality above is the same as
    \[
    \mb{E}[y_1 S_{j}(T) | \mc{G}_{T-1}] \leq \mb{E}[\mb{E}[y_1 S_{j}(T-1)| \mc{F}_{T-1}] | \mc{G}_{T-1}] \mbox{ } P\mbox{-a.e}.
    \]
    This is, by the tower property, equivalent to
    \[
    \mb{E}[y_1 S_{j}(T) | \mc{G}_{T-1}] \leq \mb{E}[y_1 S_{j}(T-1) | \mc{G}_{T-1}] \mbox{ } P\mbox{-a.e}.
    \]

    By change of measure under conditional expectation, it is enough to show
    \[
    \mb{E}[y_1 | \mc{G}_{T-1}] \mb{E}_Q[S_{j}(T) | \mc{G}_{T-1}] \leq \mb{E}[y_1 | \mc{G}_{T-1}] \mb{E}_Q[ S_{j}(T-1) | \mc{G}_{T-1}].
    \]
    This holds, because $y_1 \geq 0$ $P$-a.e. and $Q \in \mc{\bar{M}}_{I_1}^a(S, \mc{G})$.
    }
\smallskip
\item{$(ii)^*$: First, we prove this for $t=T-2$. Note that for all $A \in \mc{G}_{T-2}$, $j \in I_1$
    \[
    \begin{array}{lll}
    \int_A y_2^{T-1} S_{j}(T-1) dP &= \int_A \mb{E}[y_2^{T-1} S_{j}(T-1) | \mc{G}_{T-2}] dP \\[\smallskipamount]
                        &=\int_A \mb{E}[\mb{E}[y_1 | \mc{F}_{T-1}] S_{j}(T-1) | \mc{G}_{T-2}] dP \\[\smallskipamount]
                        &=\int_A \mb{E}[\mb{E}[y_1 S_{j}(T-1) | \mc{F}_{T-1}] | \mc{G}_{T-2}] dP \\[\smallskipamount]
                        &=\int_A \mb{E}[y_1 S_{j}(T-1)| \mc{G}_{T-2}] dP
    \end{array}
    \]
\noindent Hence, from the definition of conditional expectation and change of measure under conditional expectation, it suffices to prove
    \begin{equation}
    \label{eq: over1}
    \begin{array}{llll}
    \mb{E}[y_2^{T-2} S_{j}(T-2)| \mc{G}_{T-2}] &\geq& \mb{E}[y_1 S_{j}(T-1)| \mc{G}_{T-2}] \\[\smallskipamount]
                &=& \mb{E}[y_1 | \mc{G}_{T-2}] \mb{E}_Q[S_{j}(T-1) | \mc{G}_{T-2}].
 \end{array}
\end{equation}

    But, by the definition of $y_2^{(T-1)}$, the tower property and change of measure under conditional expectation
    \begin{equation}
    \label{eq: over2}
    \begin{array}{lll}
    \mb{E}[y_2^{T-2} S_{j}(T-2)| \mc{G}_{T-2}] = \mb{E}[y_1 S_j(T-2)| \mc{G}_{T-2}] = \mb{E}[y_1 | \mc{G}_{T-2}] \mb{E}_Q[S_j(T-2)|\mc{G}_{T-2}].
    \end{array}	
    \end{equation}

    By combining equation~(\ref{eq: over1}) and (\ref{eq: over2}), it suffices to prove that
    \[
    \mb{E}[y_1 | \mc{G}_{T-2}] \mb{E}_Q[S_j(T-2)|\mc{G}_{T-2}] \geq \mb{E}[y_1 | \mc{G}_{T-2}] \mb{E}_Q[S_{j}(T-1) | \mc{G}_{T-2}].
    \]

    This holds, since $y_1 \geq 0$ P-a.e. and $Q \in \mc{\bar{M}}_{I_1}^a(S, \mc{G})$. Similarly, one can show $(ii)^*$ for $t = 1, \ldots, T-3$.
    }
\item{$(i)^*$: Recall that $\mc{G}_0 = \{\emptyset, \Omega\}$. The inequality is trivially true for $A = \emptyset$. Hence, it only remains to check that $\mb{E}[y_2^1 S_{j}(1)] \leq \mb{E}[S_{j}(0)] = S_{j}(0)$ for $j \in I_1$. Note that
    \[
    \begin{array}{llll}
    \mb{E}[y_2^1 S_{j}(1)] &= \mb{E}[y_2^1 S_{j}(1) | \mc{G}_0] \\[\smallskipamount]
                        &= \mb{E}[\mb{E}[y_1 | \mc{F}_1] S_{j}(1) | \mc{G}_0] \\[\smallskipamount]
                        &= \mb{E}[\mb{E}[y_1 S_{j}(1) | \mc{F}_1] | \mc{G}_0] \\[\smallskipamount]
                        &= \mb{E}[y_1 S_{j}(1) | \mc{G}_0] \\[\smallskipamount]
                        &= \mb{E}_Q[S_{j}(1)] \\[\smallskipamount]
                        &= \mb{E}_Q[S_{j}(1) | \mc{G}(0)] \\[\smallskipamount]
                        &\leq S_{j}(0)
    \end{array}
    \]
    \noindent where the second equality follows from the definition of $y_2^1$ and the inequality follows from $Q \in \bar{\mc{M}}_{I_1}^a(S,\mc{G})$. Hence, $(i)^*$ holds as well.
    }
\item{The equality conditions $(i), (ii)$ and $(iii)$ follow from the same kind of arguments, based on the definition of $\bar{\mc{M}}_{I_1}^a(S,\mc{G})$ and change of measure under conditional expectation.}
\end{itemize}

Hence, any $Q \in \bar{\mc{M}}^a_{I_1}(S, \mc{G})$ corresponds to a feasible dual solution, i.e. satisfies the constraints of the dual problem~(\ref{eq: opprdualshortvilkomega2}).

\medskip

Conversely, assume there exists a feasible dual solution $y_1 \geq 0, (y_2^t)_{t=1}^{T-1}$ of problem~(\ref{eq: opprdualshortvilkomega2}).

Define $Q(F) := \int_F y_1 dP$ for all $F \in \mc{F}$. This defines a probability measure since $y_1 \geq 0$, and one can assume that $\mb{E}[y_1] = 1$ since the dual problem (\ref{eq: opprdualshortvilkomega2}) is invariant under translation. The remaining part of the proof is to show that
\begin{equation}
Q \in \bar{\mc{M}}^a_{I_1}(S, \mc{G}),
\end{equation}

\noindent i.e., that the dual feasibility conditions of problem~(\ref{eq: opprdualshortvilkomega2}) correspond to the conditions for being in $\bar{\mc{M}}^a_{I_1}(S, \mc{G})$. We divide this into several claims, which we then prove.

\smallskip

\emph{Claim 1:} $\mb{E}_Q[S_i(T) | \mc{G}_{T-1}] = \mb{E}_Q[S_i(T-1) | \mc{G}_{T-1}]$ for $i \in I_2$.

\smallskip

\emph{Proof of Claim 1:} From the definition of conditional expectation, equation $(iii)$ in problem~(\ref{eq: opprdualshortvilkomega2}) is equivalent to $\mb{E}[y_1 S_i(T) | \mc{G}_{T-1}] = \mb{E}[y_2^{T-1} S_i(T-1) | \mc{G}_{T-1}]$.
From change of measure under conditional expectation
\begin{equation}
\label{eq: claim1over1}
\mb{E}[y_1 S_i(T) | \mc{G}_{T-1}] = \mb{E}[y_1 | \mc{G}_{T-1}] \mb{E}_Q[S_i(T) | \mc{G}_{T-1}]
\end{equation}
\noindent and
\begin{equation}
\label{eq: claim1over2}
\begin{array}{lll}
\mb{E}[y_2^{T-1} S_i(T-1) | \mc{G}_{T-1}] &= \mb{E}[y_2^{T-1} | \mc{G}_{T-1}] \mb{E}_Q[S_i(T-1) | \mc{G}_{T-1}].
\end{array}
\end{equation}

\noindent By combining equations~(\ref{eq: claim1over1}) and (\ref{eq: claim1over2}), $(iii)$ is equivalent to

\[
\mb{E}[y_1 | \mc{G}_{T-1}] \mb{E}_Q[S_i(T) | \mc{G}_{T-1}]= \mb{E}[y_2^t | \mc{G}_{T-1}] \mb{E}_Q[S_i(T-1) | \mc{G}_{T-1}].
\]

By considering equation $(iii)$ for the bond and using that the market is normalized (by assumption),
\begin{equation}
\label{eq: boks}
\int_A y_1 dP = \int_A y_2^{T-1} dP \mbox{ for all } A \in \mc{G}_{T-1}.
\end{equation}
\noindent From the definition of conditional expectation, this implies that $\mb{E}[y_2^{T-1}|\mc{G}_{T-1}] = \mb{E}[y_1 | \mc{G}_{T-1}]$. Since $y_1 > 0$ a.e.,  $\mb{E}_Q[S_i(T) | \mc{G}_{T-1}] = \mb{E}[S_i(T-1)|\mc{G}_{T-1}]$. This proves Claim 1.

\medskip

\emph{Claim 2:} $\mb{E}_Q[S_i(t+k) | \mc{G}_t] = \mb{E}_Q[S_i(t)|\mc{G}_t]$ for $k \in \mb{N}$, $i \in I_2$.

\smallskip

\emph{Proof of Claim 2:} Let $i \in I_2$. First, one can show by induction that $\mb{E}_Q[S_i(T) | \mc{G}_t] = \mb{E}_Q[S_i(t)|\mc{G}_t]$ for all $t \leq T,$ $i \in I_2$, using Claim 1. Also by an inductive argument (for $i \in I_2$), this can be generalized to Claim 3.

\medskip

\emph{Claim 3:} $\mb{E}_Q[S_{j}(T) | \mc{G}_{T-1}] \leq \mb{E}_Q[S_{j}(T-1)|\mc{G}_{T-1}]$ for $j \in I_1$.

\smallskip

\emph{Proof of Claim 3:} To prove $\mb{E}_Q[S_{j}(T) | \mc{G}_{T-1}] \leq \mb{E}_Q[S_{j}(T-1)|\mc{G}_{T-1}]$ for $j \in I_1$, we use $(iii)^*$, an argument similar to that used to show Claim 1, and Claim 2.

\medskip

\emph{Claim 4:} $\mb{E}_Q[S_{j}(T) | \mc{G}_t] \leq \mb{E}[S_{j}(t)|\mc{G}_t]$ for all $t \leq T$ and $j \in I_1$.

\smallskip

\emph{Proof of Claim 4:} Let $j \in I_1$. To show that $\mb{E}_Q[S_{j}(T) | \mc{G}_t] \leq \mb{E}[S_{j}(t)|\mc{G}_t]$ for all $t \leq T$: Note that from equation~$(ii)^*$ of problem~(\ref{eq: opprdualshortvilkomega2}) for $t+1, t+2, \ldots, T-2$, it follows that
\[
\begin{array}{llll}
\int_A y_2^t S_{j}(t) dP &\geq \int_A y_2^{t+1} S_{j}(t+1) dP  &\forall \mbox{ } A \in \mc{G}_t \\
                        &\geq \int_A y_2^{t+2} S_{j}(t+2) dP  &\forall \mbox{ } A \in \mc{G}_{t+1}, \\ &&\mbox{ in particular } \forall  \mbox{ } A \in \mc{G}_t  \\
                        &\geq \ldots \\
                        &\geq \int_A y_2^{T-1} S_{j}(T-1) dP &\forall \mbox{ } A \in \mc{G}_t \\
                        &\geq \int_A y_1 S_{j}(T) dP &\forall \mbox{ } A \in \mc{G}_t
\end{array}
\]
\noindent where the final inequality uses $(iii)^*$ from problem~(\ref{eq: opprdualshortvilkomega2}). Hence, by the definition of conditional expectation and change of measure under conditional expectation
\[
\int_A \mb{E}[y_2^t | \mc{G}_t]\mb{E}_Q[S_j(t) | \mc{G}_t] dP \geq \int_A \mb{E}[y_1 | \mc{G}_t]\mb{E}_Q[S_j(T) | \mc{G}_t] dP
\]
From equation $(ii)$ for the bond, we know that $\mb{E}[y_2^t | \mc{G}_t] = \mb{E}[y_1 | \mc{G}_t]$ (see the argument related to equation~\eqref{eq: boks}), so
\begin{equation}
\label{eq: dualitycond}
\int_A \{ \mb{E}[y_2^t|\mc{G}_t] (\mb{E}_Q[S_{j}(t)|\mc{G}_t] - \mb{E}_Q[S_{j}(T) | \mc{G}_t]) \} dP \geq 0 \mbox{  } \forall \mbox{   } A \in \mc{G}_t.
\end{equation}
If $y_2^t(A) \geq 0$, but not identically equal $0$ a.e., this implies Claim 4, i.e.:
\[
\mb{E}_Q[S_{j}(t) | \mc{G}_t](A) \geq \mb{E}_Q[S_{j}(T) | \mc{G}_t](A) \mbox{ for } A \in \mc{G}_t.
\]
If $y_2^t(A) = 0$ a.e., then $Q(A) = 0$, so $\mb{E}_Q[S_{j}(T) | \mc{G}_t](A) = 0$ by convention. Hence, since the price processes are non-negative, $\mb{E}_Q[S_{j}(t)|\mc{G}_t] \geq \mb{E}_Q[S_{j}(T) | \mc{G}_t]$. This proves Claim 4.

\medskip

\emph{Claim 5:} $\mb{E}_Q[S_{j}(t+k) | \mc{G}_t] \leq \mb{E}_Q[S_{j}(t)|\mc{G}_t]$ for $k \in \mb{N}$, $j \in I_1$.

\smallskip

\emph{Proof of Claim 5:} For $A \in \mc{G}_t$ and $j \in I_1$,
\[
\begin{array}{lll}
\int_A y_2^t S_{j}(t) dP &\geq \int_A y_2^{t+k} S_{j}(t+k) dP  \\
                            &= \int_A \mb{E}[y_2^{t+k} S_{j}(t+k) | \mc{G}_t] dP \\
                            &= \int_A \mb{E}[y_2^{t+k} | \mc{G}_t] \mb{E}_Q[S_{j}(t+k) | \mc{G}_t] dP \\
                            &= \int_A \mb{E}[y_1 | \mc{G}_t] \mb{E}_Q[S_{j}(t+k) | \mc{G}_t] dP

\end{array}
\]
\noindent where the first inequality follows from $(ii)^*$ (from problem~(\ref{eq: opprdualshortvilkomega2})) iterated and the third equality from $\mb{E}[y_2^{t+k} | \mc{G}_t] = \mb{E}[y_1 | \mc{G}_t]$ (see the proof of Claim 4). Hence, by the definition of conditional expectation and since $\mb{E}[y_2^t| \mc{G}_t] = \mb{E}[y_1 | \mc{G}_t] \geq 0$ (because $y_1 \geq 0$)
\[
\int_A \{ \mb{E}[y_1 | \mc{G}_t] (\mb{E}_Q[S_{j}(t) | \mc{G}_t]  - \mb{E}_Q[S_{j}(t+k) | \mc{G}_t] \} dP \geq 0 \; \mbox{ for all } A \in \mc{G}_t.
\]
By a similar argument as for equation~(\ref{eq: dualitycond}), Claim 5 holds, i.e.,
\[
\mb{E}_Q[S_{j}(t+k) | \mc{G}_t] \leq \mb{E}_Q[S_{j}(t)| \mc{G}_t] \mbox{ for all } k \in \mb{N}, \mbox{ } j \in I_1.
\]

\medskip

By combining these claims, we see that $Q \in \bar{\mc{M}}^a_{I_1}(S, \mc{G})$, and the theorem follows.\hspace{1cm}

\end{proof}

The version of the dual problem~\eqref{eq: nydualshortvilkomega} is attractive because of its connection to martingale measures, which are an essential part of mathematical finance literature, see e.g. Karatzas and Shreve~\cite{KaratzasShreve} and {\O}ksendal~\cite{Oksendal}. Another nice feature of the formulation~\eqref{eq: nydualshortvilkomega} is that when one has found the set $\bar{\mc{M}}^a_{I_1}(S, \mc{G})$, solving the problem for each new claim $B$ may be fairly simple (depending on structure of $\bar{\mc{M}}^a_{I_1}(S, \mc{G})$) since the set does not depend on the claim. In contrast, the primal problem~\eqref{eq: forsteligning} must be solved from scratch whenever one considers a new claim $B$.

\begin{remark}
Note that Theorem~\ref{thm: omskriving} has some similarities with Theorem 1 in Kabanov and Stricker~\cite{KabanovStricker}. However, we consider the pricing problem of a contingent claim instead of the no-arbitrage criterion, which is the topic of \cite{KabanovStricker}. Moreover, we have short-selling constraints, which \cite{KabanovStricker} do not have. Also, we have a general level of partial information (not necessarily delayed) and the techniques we use, in particular the use of convex duality, are different.

Kabanov and Stricker~\cite{KabanovStricker} also comment that, to their knowledge, their proof of the partial information Dalang-Morton-Willinger theorem is the only one that does not reduce the problem to a one-step model. Our technique, using convex duality, does not rely on reduction to a one-period model either. So (to the best of our knowledge), our method of proof must be a new way to avoid reduction to one-period in discrete time models.
\end{remark}

\section{Strong duality}
\label{sec: strongduality}

The main goal of this section is to prove that there is no duality gap, i.e., that the value of the primal problem (\ref{eq: forsteligning}) is equal to the value of the dual problem (\ref{eq: nydualshortvilkomega}). This can be done using the following theorem from Pennanen and Perkki\"{o}\cite{PP} (see Theorem 9 in \cite{PP}). In order to prove strong duality, we also assume that $I_1 = \{0, 1, \ldots, N\}$, i.e. that no short-selling or borrowing is allowed.

We use the same notation as in Section~\ref{sec: shorting2}, and consider the value function $\varphi(\cdot)$ as defined in Appendix~\ref{sec: conjugate}. In the following theorem, $H$ is a stochastic process with $N+1$ components at each time $t \in \{0, 1, \ldots, T-1\}$ and $\mc{H}_{\mc{G}}$ denotes the family of all stochastic processes that are adapted to the filtration $(\mc{G}_t)_t$. Also, $F^{\infty}$ is the \emph{recession function} of $F$, defined by
 \begin{equation}
 \label{eq: huendelig}
 F^{\infty}(H(\omega), 0, \omega) := \sup_{\lambda > 0} \frac{F(\lambda H(\omega) + \bar{H}(\omega), \bar{y}(\omega), \omega) - F(\bar{H}(\omega), \bar{y}(\omega), \omega)}{\lambda}
 \end{equation}
 \noindent (which is independent of $\bar{H}, \bar{y}$). Then, we have the following theorem:

\begin{theorem}
\label{thm: PP}
(Theorem 9, Pennanen and Perkki{o}~\cite{PP}) Assume there exists $y \in Y$ and $m \in \mc{L}^1(\Omega, \mc{F},P)$ such that for $P$-a.e. $\omega \in \Omega$,

\begin{equation}
\label{eq: PP}
F(H,u) \geq  u \cdot y + m \mbox{ a.s. for all } (H, u) \in \mb{R}^{T(N+1)} \times \mb{R}^{(|I_1| + 1)T + 1},
\end{equation}
\noindent where $(\cdot)$ denotes Euclidean inner product. Assume also that
\[
A := \{H \in \mc{H}_{\mc{G}} : F^{\infty}(H, 0) \leq 0 \; P\mbox{-a.s.}\}
\]
\noindent is a linear space. Then, the value function $\varphi(u)$ is lower semi-continuous on $U$ and the infimum of the primal problem is attained for all $u \in U$.
\end{theorem}


For the proof of Theorem~\ref{thm: PP}, see \cite{PP}. Theorem~\ref{thm: PP} gives conditions for the value function $\varphi$ (see Appendix~\ref{sec: conjugate}) to be lower semi-continuous. Hence, from Theorem~\ref{thm: dualprob}, if these conditions hold, there is no duality gap since $\varphi(\cdot)$ is convex (because the perturbation function $F$ was chosen to be convex).

\begin{remark}
Note that there is a minor difference between the frameworks of Rockafellar~\cite{Rockafellar} and Pennanen and Perkki\"{o}~\cite{PP}. In the latter it is assumed that the perturbation function $F$ is a so-called convex normal integrand. However, from Example 1 in Pennanen~\cite{Pennanen} and Example 14.29 in Rockafellar and Wets~\cite{RockafellarWets}, it follows that our choice of $F$ is in fact a convex normal integrand.
\end{remark}


The following theorem states that there is no duality gap and characterizes the seller's price of the contingent claim.


\begin{theorem}
\label{thm: price2}
Consider the setting of this paper, and assume that there is no arbitrage with respect to $(\mc{G}_t)_t$. If the seller of the claim $B$ has information $(\mc{G}_t)_t$ and no short selling or borrowing is allowed, she will offer the claim at the price
\begin{equation}
\label{eq: beta}
\begin{array}{lll}
\beta := \sup_{Q \in \bar{\mc{M}}^a(S, \mc{G})} \mb{E}_Q[B].
\end{array}
\end{equation}
\noindent where $\bar{\mc{M}}^a(S, \mc{G})$ is the set the set of probability measures $Q$ on $(\Omega, \mc{F})$ that are absolutely continuous w.r.t. $P$ and are such that the price processes satisfy $E_Q[S_j(t+k) | \mc{G}_t] \leq E_Q[S_j(t) | \mc{G}_t]$ for $k \geq 0$ and $t \in 0,1, \ldots, T-k$.
\end{theorem}

\begin{proof}
 We apply Theorem~\ref{thm: PP} in order to show that there is no duality gap for our pricing problem:
\begin{itemize}
\item{We first show that the set $A$ is a linear space. We compute $F^{\infty}(H(\omega), 0, \omega)$ by choosing $\bar{y} = 0$ and $\bar{H}$ to be the portfolio that starts with $1 + \sup_{\omega \in \Omega} B(\omega)$ units of the bond and just follows the market development (without any trading) until the terminal time. Then, we find that
\[
\begin{array}{lll}
A&= \{H : \mc{G}\mbox{-adapted}, H(t) \geq 0 \mbox{ } \forall \mbox{ } t, S(t) \cdot \Delta H(t) = 0, \\[\smallskipamount]
&S(T) \cdot H(T-1) \geq 0, S(0) \cdot H(0) \leq 0\} = \{0\},
\end{array}
\]
 \noindent where the final equality holds since we assume that there is no arbitrage w.r.t. the filtration $(\mc{G}_t)_t$. Hence, $A = \{0\}$, which is a (trivial) linear space. Hence the first condition of Theorem~\ref{thm: PP} is satisfied.}

\item{To check the other assumption of the theorem, choose
\[
y = (0, (0)_t, (0)_t, -1) \in \mc{L}^q(\Omega, \mc{F}, P: \mb{R}^{(|I_1| + 1)T + 1}),
\]
\noindent where $0$ represents the $0$-function. Also, choose $m(w) = -1$ for all $\omega \in \Omega$. Then $m \in \mc{L}^1(\Omega, \mc{F}, P)$. Then, given $(H,u) \in \mb{R}^{T(N+1)} \times \mb{R}^{(|I_1| + 1)T + 1}$:

\[
\begin{array}{llll}
F(H, u) &\geq& S(0) \cdot H(0) &\mbox{ (from the definition of $F$)}\\[\smallskipamount]
&\geq& -z &\mbox{ (from the definition of $F$)}\\
&=& u \cdot y(\omega) + m(\omega) &\mbox{ (from the choice of $y$ and $m$)}. \\
\end{array}
\]}
\end{itemize}
This proves that the conditions of Theorem~\ref{thm: PP} are satisfied. Therefore, there is no duality gap, so the seller's price of the contingent claim is

\[
\sup_{Q \in \bar{\mc{M}}^a(S, \mc{G})} \mb{E}_Q[B].
\]


\end{proof}

\begin{remark}
We remark that Proposition 4.1 in F\"{o}llmer and Kramkov~\cite{FK} gives an expression for the seller's price of a claim with super-martingale conditions on the price process. However, we do not consider the same problem as \cite{FK}, since we have partial information. The presence of partial information results in a different type of martingale measure (we get a martingale- and super-martingale measure on the optional projection of the price process) than in the paper \cite{FK}.

\end{remark}

In order to prove the strong duality in Theorem~\ref{thm: price2}, we have assumed that no short-selling or borrowing is allowed. This is necessary in order for the space $A$ to be a linear space, as required by the strong duality characterization in Theorem~\ref{thm: PP}. However, we have not been able to find a numerical example where there actually is a duality gap. In the finite $\Omega$ (i.e. linear programming) case, there will be no duality gap even when borrowing or short selling is allowed. Hence, if there exists an example of a duality gap, it must be in the infinite $\Omega$ case.

This leads one to believe that it may be possible to close the duality gap in general. However, we have not found a way to achieve this through our convex analysis of the problem. Another option is to try to close the duality gap in the short selling case by analyzing the primal problem using Lagrange duality, see e.g. Pinar~\cite{Pinar1}, \cite{Pinar2}. However, as this methodology is equivalent to our convex duality approach in the discrete time setting, it seems likely that one will run into a similar problem with linearity. This is an open problem for further research.


\begin{example}
We illustrate the previous results by considering a simple numerical example. Although the results of this paper hold when $\Omega$ is an arbitrary set, we consider a situation where $\Omega$ is finite. This simplifies the intuition and allows for illustration via scenario trees.

Consider times $t=0,1,2$, $\Omega := \{\omega_1, \omega_2, \ldots, \omega_5\}$ and a market with two assets: one bank account $S_0$ and one risky asset $S_1$. Assume that the market is discounted, so $S_0(t,\omega)=1$ for all times $t$ and all $\omega \in \Omega$. Let
\[
S_1(t,\omega) := X(t,\omega) + \xi(t,\omega),
\]
\noindent i.e.  the price of the risky asset is composed of two other processes, $X$ and $\xi$. The seller does not observe these two processes, only the prices. The following scenario trees show the development of the processes $X$ and $\xi$, as well as the price development observed by the seller. Note that we only display the information needed in the following calculations.

%
 \begin{figure}[ht]
 \setlength{\unitlength}{0.7mm}
 \begin{picture}(50,80)(-40,0)
  \put(20,45){\circle*{3}} 
  \put(40,65){\circle*{3}} 
  \put(40,25){\circle*{3}} 
  \put(60,13){\circle*{3}} 
  \put(60,37){\circle*{3}} 
  \put(60,57){\circle*{3}} 
  \put(60,65){\circle*{3}} 
  \put(60,73){\circle*{3}} 
  \put(20,45){\line(1,1){20}}
  \put(20,45){\line(1,-1){20}}
  \put(40,65){\line(5,0){20}}
  \put(40,65){\line(5,2){20}}
  \put(40,65){\line(5,-2){20}}
  \put(40,25){\line(5,3){20}}
  \put(40,25){\line(5,-3){20}}

  \put(2,53){\makebox(0,0){$\Omega = \{\omega_1, \omega_2, \ldots, \omega_5\}$}}
  \put(37,75){\makebox(0,0){$\xi=3, \{\omega_1, \omega_3, \omega_5\}$}}
  \put(37,12){\makebox(0,0){$\xi=5, \{\omega_2, \omega_4\}$}}
  \put(68,73){\makebox(0,0){$\omega_1$}}
  \put(68,57){\makebox(0,0){$\omega_5$}}
    \put(68,65){\makebox(0,0){$\omega_3$}}
  \put(68,37){\makebox(0,0){$\omega_2$}}
  \put(68,13){\makebox(0,0){$\omega_4$}}

  \put(1,4){\line(35,0){80}}
  \put(20,4){\circle*{1}}
  \put(40,4){\circle*{1}}
  \put(60,4){\circle*{1}}

  \put(20,1){\makebox(0,0){$t = 0$}}
  \put(40,1){\makebox(0,0){$t = 1$}}
  \put(60,1){\makebox(0,0){$t = T = 2$}}
 \end{picture}
\caption{The process $\xi$}
\end{figure}
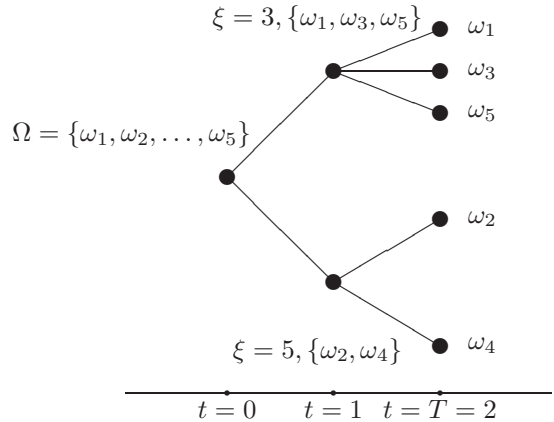

%
 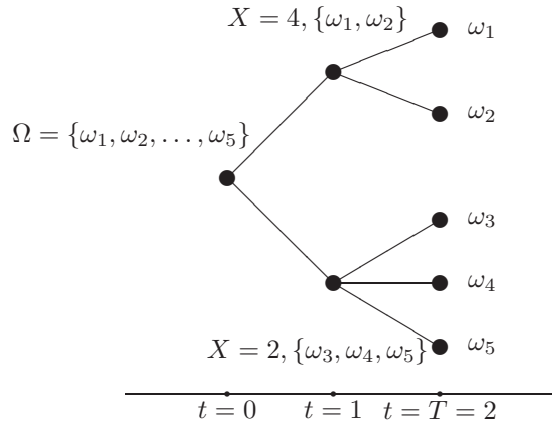
\begin{figure}[ht]
 \setlength{\unitlength}{0.7mm}
 \begin{picture}(50,80)(-40,0)
  \put(20,45){\circle*{3}} 
  \put(40,65){\circle*{3}} 
  \put(40,25){\circle*{3}} 
  \put(60,13){\circle*{3}} 
  \put(60,25){\circle*{3}} 
  \put(60,37){\circle*{3}} 
  \put(60,57){\circle*{3}} 
  \put(60,73){\circle*{3}} 
  \put(20,45){\line(1,1){20}}
  \put(20,45){\line(1,-1){20}}
  \put(40,65){\line(5,2){20}}
  \put(40,65){\line(5,-2){20}}
  \put(40,25){\line(5,3){20}}
  \put(40,25){\line(5,0){20}}
  \put(40,25){\line(5,-3){20}}

  \put(2,53){\makebox(0,0){$\Omega = \{\omega_1, \omega_2, \ldots, \omega_5\}$}}
  \put(37,75){\makebox(0,0){$X=4, \{\omega_1, \omega_2\}$}}
  \put(37,12){\makebox(0,0){$X=2, \{\omega_3, \omega_4, \omega_5\}$}}
  \put(68,73){\makebox(0,0){$\omega_1$}}
  \put(68,57){\makebox(0,0){$\omega_2$}}
  \put(68,37){\makebox(0,0){$\omega_3$}}
  \put(68,25){\makebox(0,0){$\omega_4$}}
  \put(68,13){\makebox(0,0){$\omega_5$}}

  \put(1,4){\line(35,0){80}}
  \put(20,4){\circle*{1}}
  \put(40,4){\circle*{1}}
  \put(60,4){\circle*{1}}

  \put(20,1){\makebox(0,0){$t = 0$}}
  \put(40,1){\makebox(0,0){$t = 1$}}
  \put(60,1){\makebox(0,0){$t = T = 2$}}
 \end{picture}
\caption{The process $X$}
\end{figure}

%
 \begin{figure}[ht]
 \setlength{\unitlength}{0.7mm}
 \begin{picture}(50,80)(-40,0)
  \put(20,45){\circle*{3}} 
  \put(40,65){\circle*{3}} 
  \put(40,25){\circle*{3}} 
   \put(40,45){\circle*{3}} 
 \put(60,13){\circle*{3}} 
  \put(60,32){\circle*{3}} 
  \put(60,57){\circle*{3}} 
  \put(60,73){\circle*{3}} 
 \put(60,45){\circle*{3}} 
  \put(20,45){\line(1,1){20}}
  \put(20,45){\line(1,-1){20}}
   \put(20,45){\line(1,0){20}}
 \put(40,65){\line(5,2){20}}
  \put(40,65){\line(5,-2){20}}
  \put(40,25){\line(3,1){20}}
  \put(40,45){\line(1,0){20}}
  \put(40,25){\line(5,-3){20}}

  \put(-2,53){\makebox(0,0){$S_1=6, \Omega = \{\omega_1, \omega_2, \ldots, \omega_5\}$}}
  \put(40,75){\makebox(0,0){$S_1=7, \{\omega_1, \omega_4\}$}}
  \put(40,12){\makebox(0,0){$S_1=9, \{\omega_3, \omega_5\}$}}
   \put(40,40){\makebox(0,0){$S_1=5, \{\omega_2\}$}}
  \put(73,73){\makebox(0,0){$S_1= 3, \omega_1$}}
  \put(73,56){\makebox(0,0){$S_1=8, \omega_4$}}
    \put(73,44){\makebox(0,0){$S_1=9, \omega_2$}}
  \put(73,32){\makebox(0,0){$S_1=7, \omega_3$}}
  \put(73,13){\makebox(0,0){$S_1=4, \omega_5$}}

  \put(1,4){\line(35,0){80}}
  \put(20,4){\circle*{1}}
  \put(40,4){\circle*{1}}
  \put(60,4){\circle*{1}}

  \put(20,1){\makebox(0,0){$t = 0$}}
  \put(40,1){\makebox(0,0){$t = 1$}}
  \put(60,1){\makebox(0,0){$t = T = 2$}}
 \end{picture}
\caption{The price process $S_1$}
\end{figure}
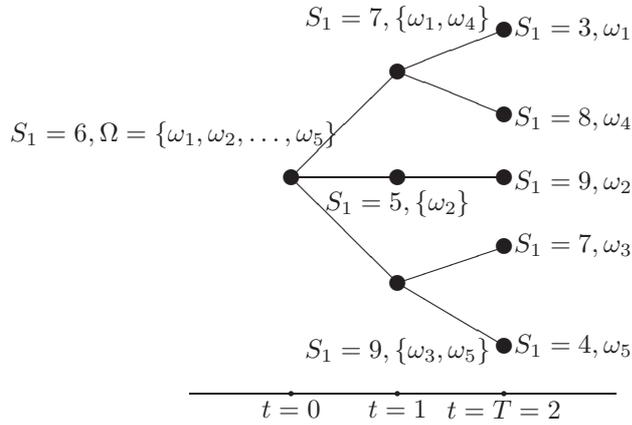
Full information in this market corresponds to observing both processes $X$ and $\xi$, i.e. the full information filtration $(\mc{F}_t)_t$ is the sigma algebra generated by $X$ and $\xi$, $\sigma(X,\xi)$. However, the filtration observed by the seller $(\mc{G}_t)_t$, generated by the price processes, is (strictly) smaller than the full information filtration. For instance, if you observe that $\xi(1)=3$ and $X(1)=4$, you know that the realized scenario is $\omega_1$. However, this is not possible to determine only through observation of the price process $S_1$. Hence, this is an example of a model with hidden processes, which is a kind of partial information that is not delayed information.

Assume that the seller is not allowed to short-sell. In this case, the seller's problem \eqref{eq: forsteligning} is to solve the following minimization problem w.r.t. $v$ and all $(\mc{G}_t)_t$-adapted trading strategies $H$:

\begin{equation}
\label{eq: LP_example}
\begin{array}{lrlll}
\inf_{v, H} &v \\[\smallskipamount]
\mbox{s.t.} \\[\smallskipamount]
&3H_1(1,\omega_1) + H_0(1,\omega_1) &\geq& B(\omega_1) \\[\smallskipamount]
&9H_1(1,\omega_2) + H_0(1,\omega_2) &\geq& B(\omega_2) \\[\smallskipamount]
&5H_1(1,\omega_3) + H_0(1,\omega_3) &\geq& B(\omega_3) \\[\smallskipamount]
&8H_1(1,\omega_1) + H_0(1,\omega_1) &\geq& B(\omega_4) \\[\smallskipamount]
&4H_1(1,\omega_3) + H_0(1,\omega_3) &\geq& B(\omega_5) \\[\smallskipamount]
&H_0(1,\omega_1) - H_0(0) + 7\big( H_1(1,\omega_19 - H_1(0)) \big) &=& 0 \\[\smallskipamount]
&H_0(1,\omega_3) - H_0(0) + 5\big( H_1(1,\omega_3) - H_1(0)) \big) &=& 0 \\[\smallskipamount]
&H_0(1,\omega_2) - H_0(0) + 9\big( H_1(1,\omega_2) - H_1(0)) \big) &=& 0 \\[\smallskipamount]
&H_j(0) \geq 0, H_j(1, \omega_i) &\geq& 0 \mbox{ for } j=0,1, i=1,2,3 \\[\smallskipamount]
&H_0(0) + 6H_1(0) &\leq& v
\end{array}
\end{equation}
\noindent where $H_1(1,\omega_1) = H_1(1,\omega_4)$ and $H_1(1,\omega_3) = H_1(1,\omega_5)$ due to $H$ being $(\mc{G}_t)_t$-adapted. This is a linear programming problem which can be solved using the simplex algorithm. Note that the simplex algorithm is a duality method, which leads to a dual problem equivalent to the one we have derived in Section~\ref{sec: main}. An advantage with solving this problem directly is that we get the trading strategy $H$ explicitly. However, a downside with solving problem~\eqref{eq: LP_example}  directly is that for each new claim $B$, the problem must be solved from scratch. This is not the case when solving the dual problem instead. From Theorem~\ref{thm: price2}, the dual problem for the price of the claim is:

\begin{equation}
\label{eq: dualproblem_eksempel}
\sup_{Q \in \bar{\mc{M}}^a(S,\mc{G})} \mb{E}_Q[B]
\end{equation}

\noindent where $\bar{\mc{M}}^a(S,\mc{G})$ is the set of absolutely continuous probability measures making the price process $S_1$ a $(\mc{G}_t)_t$-conditional super-martingale (and $S_0$, but this is trivial since the price processes are discounted). In order to solve this problem, we must find the set $\bar{\mc{M}}^a(S,\mc{G})$. By using the definition of $\bar{\mc{M}}^a(S,\mc{G})$, we get a system of linear inequalities to solve. By solving these, using for example Fourier Motzkin elimination, we find that

\begin{equation}
\begin{array}{lll}

\label{eq: mengde}
 \bar{\mc{M}}_1^a(S,\mc{G}) = \{Q = (q_1,q_2, \ldots, q_5) : 0 \leq q_3 \leq \frac{6}{21}, 0 \leq q_4 \leq 1-q_3, \\[\smallskipamount]
 0 \leq q_5 \leq 1- q_3 - q_4, 0 \leq q_1 \leq 1 - q_3-q_4-q_5 \mbox{ and } q_2=1-q_1-q_3-q_4-q_5\}
\end{array}
 \end{equation}

Hence, given some claim $B$, one can solve the problem \eqref{eq: dualproblem_eksempel} for the set in \eqref{eq: mengde} in order to find the seller's price. When one would like to find prices for several claims $B_1, B_2, \ldots, B_m$, solving the dual problem is simpler than solving the primal LP problem since the set $ \bar{\mc{M}}^a(S,\mc{G})$ is the same for all the claims.

\end{example}

\section{Conclusions}
\label{sec: finalremarks}

In this paper, we have shown how convex duality can be used to obtain pricing results for a seller of a claim who has partial information and is facing short selling constraints in a discrete time financial market model. This gives new results, which are summarized in Theorem~\ref{thm: omskriving} and Theorem~\ref{thm: price2}.

It seems natural that these results can be generalized to a model with continuous time, possibly using a discrete time approximation. However, this may be quite technical.

\appendix

\section{Conjugate duality and paired spaces}
\label{sec: conjugate}

Conjugate duality theory (also called convex duality), introduced by Rockafellar~\cite{Rockafellar}, provides a method for solving very general optimization problems via dual problems.

Let $X$ be a linear space, and let $f: X \rightarrow \mathbb{R}$ be a function. The minimization problem $\min_{x \in X} f(x)$ is called the \emph{primal problem}, denoted $(P)$. In order to apply the conjugate duality method to the primal problem, we consider an abstract optimization problem $\min_{x \in X} F(x,u)$ where $F: X \times U \rightarrow \mathbb{R}$\index{$F(x,u)$} is a function such that $F(x,0) = f(x)$, $U$ is a linear space and $u \in U$ is a parameter chosen depending on the particular problem at hand. The function $F$ is called the \emph{perturbation function}. We would like to choose $(F,U)$ such that $F$ is a closed, jointly convex function of $x$ and $u$.

Corresponding to this problem, one defines the \emph{optimal value function}
\begin{equation}
 \varphi(u)\index{$\varphi(\cdot)$} := \inf_{x \in X} F(x,u) \hspace{0,05cm}, \hspace{0,3cm} u \in U.
\end{equation}
Note that if the perturbation function $F$ is jointly convex, then the optimal value function $\varphi(\cdot)$ is convex as well.

A \emph{pairing} of two linear spaces $X$ and $V$ is a real-valued bilinear form $\langle \cdot , \cdot \rangle$ on $X \times V$. Assume there is a pairing between the spaces $X$ and $V$. A topology on $X$ is \emph{compatible} with the pairing if it is a locally convex topology such that the linear function $\langle \cdot , v \rangle$ is continuous, and any continuous linear function on $X$ can be written in this form for some $v \in V$. A compatible topology on $V$ is defined similarly. The spaces $X$ and $V$ are \emph{paired spaces} if there is a pairing between $X$ and $V$ and the two spaces have compatible topologies with respect to the pairing. An example is the spaces $X = L^{p}(\Omega, F, P)$ and $V = L^{q}(\Omega, F, P)$, where $\frac{1}{p} + \frac{1}{q} = 1$. These spaces are paired via the bilinear form $\langle x , v \rangle = \int_{\Omega} x(s) v(s) dP(s)$.

In the following, let $X$ be paired with another linear space $V$, and $U$ paired with the linear space $Y$. The choice of pairings may be important in applications.  Define the \emph{Lagrange function} $K: X \times Y \rightarrow \bar{\mathbb{R}}$ to be $K(x,y) := \inf\{F(x,u) + \langle u,y \rangle : u \in U\}$. The following Theorem~\ref{thm: Lagrange} is from Rockafellar~\cite{Rockafellar} (see Theorem 6 in \cite{Rockafellar}).
\begin{theorem}
 \label{thm: Lagrange}
The Lagrange function $K$ is closed, concave in $y \in Y$ for each $x \in X$, and if $F(x,u)$ is closed and convex in $u$
\begin{equation}
 \label{eq:box}
f(x) = \sup_{y \in Y} K(x,y).
\end{equation}
\end{theorem}
For the proof of this theorem, see Rockafellar~\cite{Rockafellar}. Motivated by Theorem~\ref{thm: Lagrange}, we define the \emph{dual problem} of $(P)$,
\begin{eqnarray}
 (D) \hspace{1cm} \max_{y \in Y} g(y) \nonumber
\end{eqnarray}
\noindent where $g(y) := \inf_{x \in X} K(x,y)$.

One reason why problem $(D)$ is called the dual of the primal problem $(P)$ is that, from equation (\ref{eq:box}), problem $(D)$ gives a lower bound on problem $(P)$. This is called \emph{weak duality}. Sometimes, one can prove that the primal and dual problems have the same optimal value. If this is the case, we say that there is \emph{no duality gap} and that \emph{strong duality holds}. The next theorem (see Theorem 7 in Rockafellar~\cite{Rockafellar}) is important:
\begin{theorem}
 \label{thm: dualprob}
The function $g$ in $(D)$ is closed and concave. Also
\begin{eqnarray}
\sup_{y \in Y} g(y) = \cl (\co (\varphi))(0)  \nonumber
\end{eqnarray}
and
\begin{eqnarray}
\inf_{x \in X} f(x) = \varphi(0). \nonumber
\end{eqnarray}
\end{theorem}
\noindent (where $\cl$ and $\co$ denote respectively the closure and the convex hull of a function, see Rockafellar~\cite{RockafellarConvex}). For the proof, see Rockafellar~\cite{Rockafellar}. Theorem~\ref{thm: dualprob} implies that \emph{if the value function $\varphi$ is convex, the lower semi-continuity of $\varphi$ is a sufficient condition for the absence of a duality gap}.



\end{document}